\documentclass[a4paper,12pt]{article}

\usepackage{amsmath,amssymb,amsthm,bbm,graphicx,url}
\usepackage[round]{natbib}

\theoremstyle{plain}  
\newtheorem{thm}{Theorem}[section]

\theoremstyle{definition} 
\newtheorem{defn}{Definition}[section]

\theoremstyle{remark}

\newcommand{\real}{\mathbb{R}} 

\newcommand{\E}{\mathbb{E}} 
\newcommand{\Q}{\mathbb{Q}} 
\newcommand{\EQ}{\E_{\hspace{0.2mm}\Q}} 

\newcommand{\cA}{\mathcal{A}} 
 
\newcommand{\cE}{\mathcal{E}} 
\newcommand{\cK}{\mathcal{K}} 
\newcommand{\cL}{\mathcal{L}} 

\newcommand{\one}{\mathbbm{1}}

\newcommand{\prob}{\mathrm{pr}}

\begin{document}

\title{Copula Calibration}
\author{Johanna F.~Ziegel \and Tilmann Gneiting}
\date{}
\maketitle

\begin{abstract}

We propose notions of calibration for probabilistic forecasts of
general multivariate quantities.  Probabilistic copula calibration is
a natural analogue of probabilistic calibration in the univariate
setting.  It can be assessed empirically by checking for the
uniformity of the copula probability integral transform (CopPIT),
which is invariant under coordinate permutations and coordinatewise
strictly monotone transformations of the predictive distribution and
the outcome.  The CopPIT histogram can be interpreted as a
generalization and variant of the multivariate rank histogram, which
has been used to check the calibration of ensemble forecasts.
Climatological copula calibration is an analogue of marginal
calibration in the univariate setting.  Methods and tools are
illustrated in a simulation study and applied to compare raw numerical
model and statistically postprocessed ensemble forecasts of bivariate
wind vectors.

\end{abstract}

\section{Introduction}  \label{sec:intro}

The past two decades have witnessed major developments in the
scientific approach to forecasting, in that probabilistic forecasts,
which take the form of probability distributions over future
quantities and events, have been replacing single-valued point
forecasts in a wealth of applications \citep{GKatzfuss2014}.  The goal
in probabilistic forecasting is to maximize the sharpness of the
predictive probability distributions subject to calibration
\citep{GneitingETAL2007}.  Calibration concerns the statistical
compatibility between the predictive distributions and the realizing
observations; in a nutshell, the observations are supposed to be
indistinguishable from random numbers drawn from the predictive
distributions.

For probabilistic forecasts of univariate quantities various types of
calibration have been established \citep{GRanjan2013}.  In particular,
a forecast is probabilistically calibrated if its probability integral
transform (PIT), i.e., the value of the predictive cumulative
distribution function at the realizing observation, is uniformly
distributed.  Accordingly, empirical checks for the uniformity of
histograms of PIT values have formed a cornerstone of density forecast
evaluation \citep{Dawid1984, DieboldETAL1998, GneitingETAL2007}.

In this paper we introduce notions of calibration for probabilistic
forecasts of multivariate quantities and propose tools for empirical
calibration checks in such settings, as recently called for in
hydrologic and meteorological applications \citep{SchaakeETAL2010,
Pinson2013, SchefzikETAL2013}.  In Section \ref{sec:copPIT} we study a
natural multivariate extension of the univariate PIT that is invariant
under coordinate permutations and coordinatewise strictly monotone
transformations of the predictive distribution and the realizing
observation, namely, the copula probability integral transform
(CopPIT).  Probabilistic copula calibration can be assessed
empirically by checking the uniformity of the CopPIT histogram, which
can be viewed as a generalization and variant of the multivariate rank
histogram proposed by \citet{GneitingETAL2008}.  Furthermore, we
introduce the notion of climatological copula calibration, which is an
analogue of marginal calibration in the univariate setting.  The
strengths of these notions and tools include their ease of
interpretability and their applicability to both density and ensemble
forecasts.

In Section \ref{sec:simulation} we employ CopPIT histograms in a
simulation study, and in Section \ref{sec:EMOS} we use them to compare
raw numerical model and statistically postprocessed ensemble forecasts
of bivariate wind vectors over Germany.  The paper ends with a
discussion in Section \ref{sec:discussion}.

\section{Multivariate notions of calibration}  \label{sec:copPIT} 

We introduce the copula probability integral transform (CopPIT) and
the notions of probabilistic copula calibration and climatological
copula calibration within the prediction space setting of
\cite{GRanjan2013}.  Throughout, we identify a probability measure on
$\real^d$ with its cumulative distribution function (CDF).

The Kendall distribution function $\cK_H$ of a probability measure or CDF $H$
on $\real^d$ is defined as 
\[
\cK_H(w) = \prob \{ H(X) \leq w \} 
\quad \mbox{for} \quad 
w \in [0,1],
\]
where the random vector $X$ has distribution $H$.  It is well known
that if $d = 1$ and $H$ is continuous then $\cK_H$ corresponds to a
uniform distribution on $[0,1]$.  In dimension $d > 1$, the Kendall
distribution depends only on the copula of the probability measure $H$
and generally it is not uniform \citep{BarbeETAL1996}.  In fact, for
any CDF $K$ on $[0,1]$ with $K(w) \geq w$ for $w \in [0,1]$ and any
integer $d > 1$, there exists a probability measure $H$ on $\real^d$
such that $\cK_H = K$ \citep{NelsenETAL2003, GenestETAL2011}.

\subsection{Probabilistic and climatological copula calibration}  \label{sec:notions} 

As noted, we work in the prediction space setting introduced by
\citet{GRanjan2013}.  Specifically, let $(\Omega, \cA, \Q)$ be a
probability space.  Let $Y$ be an $\real^d$-valued random vector on
$\Omega$, and let $H$ be a $d$-variate CDF-valued random quantity that
is measurable with respect to some sub $\sigma$-algebra $\cA_0
\subseteq \cA$.  Furthermore, let the random variable $V$ be uniformly
distributed on the unit interval $[0,1]$ and independent of $Y$ and
$\cA_0$.

The CDF-valued random quantity $H$ provides an $\cA_0$-measureable
predictive probability measure for the $\real^d$-valued outcome $Y$.
It is said to be ideal relative to $\cA_0$ if it equals the
conditional law of $Y$ given $\cA_0$, which we denote by $H =
\cL(Y|\cA_0)$.  Thus, an ideal forecast honors the information in the
sub $\sigma$-algebra $\cA_0 \subseteq \cA$ to the full extent
possible. For a function $f$ on the real line, we use the notation $f(y-)=\lim_{x \uparrow y} f(x)$ to denote the left-hand limit, if it exists.

\begin{defn}[CopPIT]
In the prediction space setting, the random variable
\begin{equation}\label{eq:CopPIT}
U_H = \cK_H\{H(Y)-\} + V \left[ \cK_H\{H(Y)\} - \cK_H\{H(Y)-\} \right]
\end{equation}
is the copula probability integral transform (CopPIT) of the
CDF-valued random quantity $H$.
\end{defn}

If $T$ is a deterministic coordinatewise strictly monotone transformation 
on $\real^d$, i.e., 
\[
T(x_1, \ldots, x_d) = \left( T_1(x_1), \ldots, T_d(x_d) \right)
\]
where the mappings $T_1, \ldots, T_d$ are real-valued and strictly
increasing, the distribution of $U_H$ for the probabilistic forecast
$H$ and the outcome $Y$ is the same as that of $U_{H \circ T^{-1}}$
for the probabilistic forecast $H \circ T^{-1}$ and the outcome
$T(Y)$.  The distribution of $U_H$ also is invariant under coordinate
permutations.  An interesting open question is for the largest
class of transformations under which this invariance holds, with the
class of the locally orientation preserving functions being a 
candidate.

\begin{defn}
The forecast $H$ is probabilistically copula calibrated if its CopPIT
is uniformly distributed on the unit interval.
\end{defn}

Probabilistic copula calibration can be viewed as a multivariate
generalization of the notion of probabilistic calibration in the
univariate case.  In the prediction space setting, let $F$ be a
univariate CDF-valued random quantity for the real-valued outcome $Y$.
\citet[Definition 2.6]{GRanjan2013} define $F$ to be probabilistically
calibrated if
\begin{equation}  \label{eq:PIT}
U_F = F(Y-) + V\{F(Y) - F(Y-)\}
\end{equation}
is standard uniformly distributed.  If the dimension is $d = 1$ then equation 
\eqref{eq:CopPIT} is the same as equation \eqref{eq:PIT}.

\begin{defn}
The forecast $H$ is climatologically copula calibrated if 
\begin{equation}\label{eq:clical}
\Q \{ H(Y) \leq w \} = \EQ \{ \cK_H(w) \}
\quad \mbox{for} \quad 
w \in [0,1].
\end{equation}
\end{defn}

The concept of climatological copula calibration can be interpreted as
marginal calibration of the Kendall distribution, where marginal
calibration refers to the univariate prediction space setting, as
follows \citep[Definition 2.6]{GRanjan2013}.  If $F$ is a univariate
CDF-valued random quantity for the real-valued outcome $Y$, then it is
marginally calibrated if $\Q (Y \leq y) = \EQ \{ F(y) \}$ for $y \in
\real$.

The following result justifies the quest for probabilistically and
climatologically copula calibrated probabilistic forecasts in
practical settings.

\begin{thm}
If the forecast $H$ is ideal with respect to the $\sigma$-algebra
$\cA_0$, then it is both probabilistically and
climatologically copula calibrated.
\end{thm}

\begin{proof}
Suppose that $H = \cL(Y|\cA_0)$ and let $w \in [0,1]$. Then
\[ 
\Q \{ H(Y) \leq w \} 
= \EQ \, \EQ \left[ \one \{ H(Y) \leq w \} | \cA_0 \right] 
= \EQ \{ \cK_H(w) \}, 
\]
whence $H$ is climatologically copula calibrated. Turning to
probabilistic copula calibration, well known results for non-random
CDFs and conditional expectations imply that $\Q \{ U_H \leq w \} =
w$.
\end{proof}

Suppose that the probabilistic forecasts $F_1, \ldots, F_d$ for the
marginals of the random vector $Y = (Y_1,\dots,Y_d)$ are
probabilistically calibrated.  Then probabilistic copula calibration
can be seen as a property that depends only on the copula $C$ of the
forecast $H$ and the copula of the outcome vector $Y$, as follows.
Probabilistic calibration of the marginals implies that the random
vector $W = \left( U_{F_1}, \ldots, U_{F_d} \right)$ has uniformly
distributed marginals.  Therefore, the problem of predicting $Y$ by
$H$ can be reduced to predicting $W$ by a copula $C$.  Then
\[
H = C \circ \left( F_1, \ldots, F_d \right)
\]
yields a multivariate probabilistic forecast of $Y$ with
probabilistically calibrated marginals. For a related discussion 
in the context of ensemble forecasts, see \citet{SchefzikETAL2013}.

\subsection{Empirical assessment of copula calibration}  \label{sec:empirical}

In the practice of forecast evaluation, one observes a sample 
\[
(H_1,y_1), \ldots, (H_J,y_J)
\]
from the joint distribution of the probabilistic forecast and the
outcome.

To assess probabilistic copula calibration one can plot a histogram of
the empirical CopPIT values 
\begin{equation} \label{eq:CopPIT1} 
u_j = \cK_{H_j}\{H_j(y_j)-\} + v_j \, [ \cK_{H_j}\{H_j(y_j)\} -
\cK_{H_j}\{H_j(y_j)-\} ] 
\end{equation} 
for $j = 1, \ldots, J$, where $v_1, \ldots, v_J$ are independent
standard uniformly distributed random numbers.  Based on ideas in
\cite{CzadoETAL2009}, one can also define a non-randomized version of
the CopPIT, but we do not pursue this here.  In most cases of
practical interest, the Kendall distribution is continuous and then we
can write 
\begin{equation} \label{eq:CopPIT2} u_j =
\cK_{H_j}\{H_j(y_j)\}, 
\end{equation} 
without any need to invoke $v_j$.  If $d = 1$, the CopPIT histogram
coincides with the PIT histogram, the key tool in checking the
calibration of univariate probabilistic forecasts
\citep{DieboldETAL1998, GneitingETAL2007, CzadoETAL2009}.  If the
forecasts are probabilistically copula calibrated, the CopPIT histogram
is uniform up to random fluctuations, and deviations from uniformity
can be interpreted diagnostically, as illustrated in Section
\ref{sec:simulation}.

For multivariate distributions with an Archi\-medean copula the
Kendall distribution function $\cK_H$ is available in closed form
\citep{McNeilNeslehova2009}, and then we can readily evaluate
\eqref{eq:CopPIT1} or \eqref{eq:CopPIT2}.  For other types of
distributions, we approximate $\cK_H$ by the empirical CDF of $H(x_1),
\ldots, H(x_n)$ for some large $n$, where $x_1, \ldots, x_n$ is a
sample from a $d$-variate population with CDF $H$.  Another
approximation that does not require the potentially costly evaluation
of $H$ uses the empirical Kendall distribution function $\cK_n$, i.e.,
the empirical CDF of the pseudo-observations
\begin{equation}  \label{eq:empK}
w_k = \frac{1}{n}\sum_{j=1}^n \one\{ x_j \preceq x_k\}
\qquad \text{for $k = 1, \ldots, n$}, 
\end{equation}
where $x_j = (x_{j1}, \ldots, x_{jd}) \preceq x_k =(x_{k1}, \dots,
x_{kd})$ if $x_{jl} \leq x_{kl}$ for $l = 1, \ldots, d$.  As
\cite{BarbeETAL1996} show, the empirical Kendall distribution
function $\cK_n$ generally converges to $\cK_H$.

To assess climatological copula calibration one can plot
\[
\frac{1}{J} \sum_{j=1}^J \one \{ H_j(y_j) \leq w \} 
\quad \text{vs.} \quad 
\frac{1}{J} \sum_{j=1}^J \cK_{H_j}(w) 
\]
for $w \in [0,1]$, which are the empirical analogues of the left- and
right-hand sides of \eqref{eq:clical}.  If the forecasts are
calibrated the resulting plot ought to be close to the diagonal.

\subsection{Comparison to the multivariate rank histogram}  \label{sec:comparison}

As noted, the CopPIT histogram generalizes the multivariate rank
histogram introduced by \citet{GneitingETAL2008} in the context of
ensemble forecasts.  This refers to the situation in which the
probabilistic forecasts $H_1, \ldots, H_J$ are empirical measures with
a fixed size $m$.

For ease of exposition, we drop the indices and suppose that the
forecast $H$ places mass $1/m$ at each of $x_1, \ldots, x_m \in
\real^d$, while the outcome is $y \in \real^d$.  The associated
multivariate rank is obtained as follows.  Define pre-ranks $\rho_0
= 1 + \sum_{i=1}^m \one(x_i \preceq y)$ and 
\[
\rho_k = \one(y \preceq x_k) + \sum_{i=1}^m \one(x_i \preceq x_k) \quad \text{for $k = 1, \ldots, m$.}
\]
The multivariate rank then is the rank of the
observation pre-rank $\rho_0$ among $\rho_0, \rho_1, \ldots, \rho_m$,
with ties resolved at random.  Conditional on $H$ and $y$ we thus get
a multivariate rank with a discrete uniform distribution on the integers 
\begin{equation}  \label{eq:rank}   
1 + \sum_{k=1}^m \one ( \rho_k < \rho_0 ), \; \ldots \; , 
1 + \sum_{k=1}^m \one ( \rho_k \leq \rho_0 ).
\end{equation} 
We now link the multivariate rank and the CopPIT.  If $H$ is the
empirical measure with mass $1/m$ at $x_1, \ldots, x_m \in \real^d$,
its Kendall distribution function can be expressed in terms of the
pseudo-observations at \eqref{eq:empK}, in that
\[
\cK_H(w) = \frac{1}{m} \sum_{k=1}^m \one(w_k \le w)
\quad \text{for} \quad
 w \in [0,1]. 
\]
Since $\rho_0 = m H(y) + 1$ and $\rho_k = m w_k + \one(y \preceq x_k)$
for $k = 1, \ldots, m$, we can express the CopPIT value
\eqref{eq:CopPIT1} in terms of the pseudo-ranks.  A bit of algebra shows
that conditional on $H$ and $y$ the CopPIT value has a uniform
distribution on the interval
\begin{equation}  \label{eq:CopPIT.rank} 
\left[ 
\frac{1}{m} \sum_{k=1}^m \one \{ \rho_k - \one(y \preceq x_k) < \rho_0 - 1 \},  
\frac{1}{m} \sum_{k=1}^m \one \{ \rho_k - \one(y \preceq x_k) \leq \rho_0 - 1 \} \right] \! .
\end{equation} 
A comparison of \eqref{eq:rank} and \eqref{eq:CopPIT.rank} suggests
that if the ensemble size $m$ is large the CopPIT and the multivariate
rank histogram tend to look nearly identical.  If $m$ is small this
may not be the case, as we illustrate in Section \ref{sec:EMOS}.

The multivariate rank histogram has also been used to assess the
calibration of probabilistic forecasts in the form of continuous
multivariate distributions.  \cite{SchuhenETAL2012} transform
predictive densities for bivariate wind vectors into ensemble
forecasts, by drawing a simple random sample from each predictive
distribution, where the particular choice of the sample size $m = 8$
allows for a better comparison with the underlying ensemble forecast.
In such settings we prefer to work with the CopPIT histogram, as it
makes better use of the structure of the predictive distributions and
does not induce additional randomness into the evaluation procedure.

We illustrate this latter aspect in a simulation setting in dimension
$d = 50$, where we choose the sample size $m = 8$ to compute the
multivariate rank histograms.  In weather and climate forecasting,
ensemble systems operate with small $m$ and very high $d$
\citep{GRaftery2005, Leutbecher2008}, so this scenario is practically
relevant.  Specifically, let $B_1$ and $B_2$ be independent beta
variables with parameters $(\alpha_1,\beta_1) = (2,5)$ and
$(\alpha_2,\beta_2) = (5,2)$.  Conditional on $(B_1,B_2)$ the outcome
vector has a Frank copula with each pairwise Kendall's $\tau$ equal to
$(B_1+B_2)/2$.  The forecast copula is either the true Frank copula, a
Frank copula with each pairwise $\tau$ equal to $0.8 (B_1+B_2)/2$, or
a Joe copula with each pairwise $\tau$ equal to $(B_1+B_2)/2$, as
described by \citet{Nelsen2006}.  The 50 marginals are all standard
normal and correctly predicted.

\begin{figure}[t]
\centering
\includegraphics{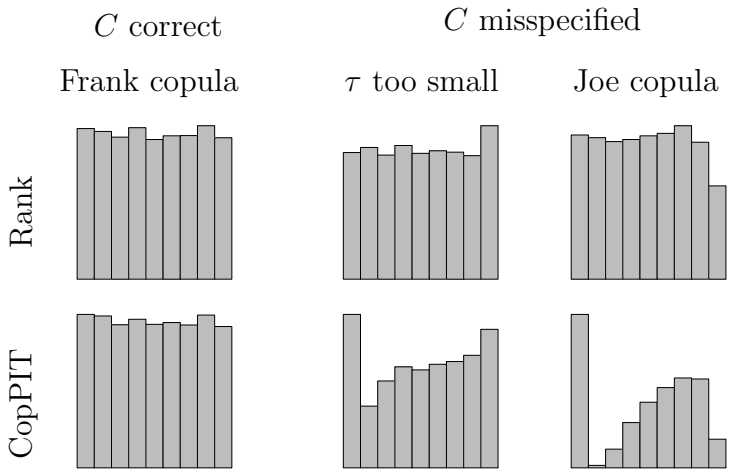}
\caption{Multivariate rank and CopPIT histograms in the high-dimensional simulation 
  setting described in the text.  \label{fig:highdim}}
\end{figure}

Figure \ref{fig:highdim} shows multivariate rank and CopPIT histograms
in this setting, based on a sample of $4,000$ forecast--observation
pairs.  The rank histograms have difficulties in detecting the
deficient probabilistic forecasts due to the aforementioned
discretization effect.  In contrast, the CopPIT histograms for the
forecasts with the misspecified copulas are non-uniform, as desired.

\subsection{Directional copula calibration}  \label{sec:directional}

The CopPIT is a natural multivariate generalization of the PIT in the
univariate setting.  We now discuss a further generalization that
allows for directional approaches.  In doing so, we refer to the
probabilistic forecast for the $\real^d$-valued outcome $Y$ either by
$H$ or $\mu$, with $H$ denoting a CDF and $\mu$ the associated
probability measure.

Let $e_1, \ldots, e_d$ be an orthonormal basis of $\real^d$ and let
$\cE$ be the closed convex cone spanned by this basis.  We define the
$\cE$-CDF of the probability measure $\mu$ as
\[
H^\cE : \real^d \to [0,1], \qquad x \mapsto \mu(x + \cE).
\]
Any function $H^{\cE}$ characterizes the probability measure $\mu$.
The usual CDF is obtained by choosing $e_j = (e_{1j}, \ldots, e_{dj})$
with $e_{ij} = - \one(i=j)$, whereas the survival function of $\mu$ is
$H^{\cE}$ with $e_{ij} = \one(i=j)$.  The CopPIT depends on the
particular CDF chosen, and distinct choices of $\cE$ may reveal
distinct facets of calibration or the lack thereof.  In principle, one
could envision a procedure in the style of a projection pursuit
algorithm \citep{Huber1985} that finds those $\cE$ where the deviation
of the CopPIT histogram from uniformity is the most pronounced.  In the
case of density forecasts a related idea was considered by
\citet{Ishida2005}.

\begin{figure}[t]
\centering
\includegraphics{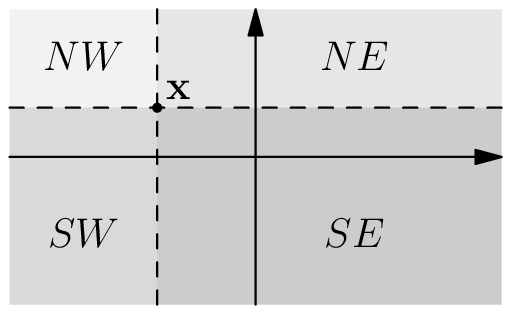}
\quad
\includegraphics{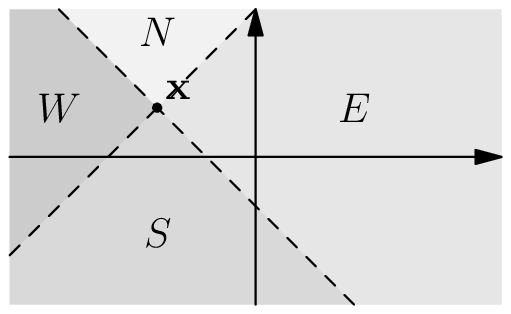}
\caption{Illustration of quadrants for directional
  CopPITs \label{fig:quadrants}}
\end{figure}

Certain choices of the cone $\cE$ might be particularly useful.  We
illustrate this for $d=2$, but the idea generalizes to higher
dimensions.  Let SW be the convex cone spanned by $(-1,0)$ and
$(0,-1)$, i.e., the south-west quadrant.  Analogously we define the
quadrants SE, NE, and NW, as illustrated in Figure
\ref{fig:quadrants}.  If the mariginals are probabilistically
calibrated, probabilistic copula calibration with respect to $H^{\rm
SW}$, which is the classical multivariate CDF, only depends on the
forecast copula.  This argument remain valid for $H^{\rm SE}$, $H^{\rm
NE}$, and $H^{\rm NW}$, with the latter being the multivariate
survival function.

Similarly, we can assess directional climatological copula calibration
by plotting 
\[
\frac{1}{J} \sum_{j=1}^J \one \{ H^\cE_j(y_j) \leq w \} 
\quad \text{vs.} \quad 
\frac{1}{J} \sum_{j=1}^J \cK_{H^\cE_j}(w) 
\]
for $w \in [0,1]$ and suitable choices of the cone $\cE$.

\section{Simulation study}  \label{sec:simulation} 

\begin{table}[p]
\caption{Parameters of forecast distributions in the simulation study \label{tab:1}}
\medskip
\centering
\begin{tabular}{c|ccc}
\hline
\hline
Forecast & First Margin $F_1$   & Second Margin $F_2$        & Copula $C$ \\
\hline
T & correct                     & correct                    & correct                  \\
  & $\mu_1 = 2 - B_1$           & $\sigma_2^2 = 1/B_2$       & $\tau = (B_1+B_2)/2$ \\ 
[1ex] 
F & biased                      & underdispersed             & misspecified \\
  & $\hat\mu_1 = 0.8 (2 - B_1)$ & $\hat\sigma_2^2 = 0.8/B_2$ & $\hat\tau = 0.6 (B_1+B_2)/2$ \\
\hline
\end{tabular}
\bigskip
\end{table}

\begin{figure}[p]
\centering
\includegraphics{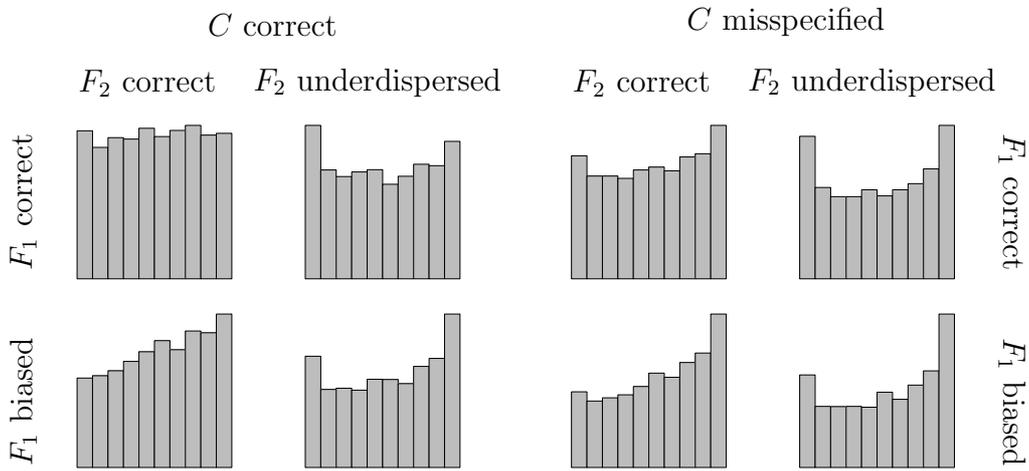}
\caption{CopPIT histograms for the forecasters in the simulation
  study \label{fig:2}}
\end{figure}

We consider the following simulation setting in dimension $d = 2$.
Let $B_1$ and $B_2$ be independent beta variables with parameters
$(\alpha_1,\beta_1) = (2,5)$ and $(\alpha_2,\beta_2) = (5,2)$,
respectively.  Conditional on $(B_1,B_2)$ the outcome vector $Y =
(Y_1,Y_2)$ has normal margins and a Gumbel copula with Kendall's
$\tau$ equal to $(B_1+B_2)/2$, as described by \citet{Nelsen2006}.
The margin $Y_1$ has mean $\mu_1 = 2 - B_1$ and unit variance; the
margin $Y_2$ has mean zero and variance $\sigma_2^2 = 1/B_2$.

We assess eight probabilistic forecasters with various types of
forecast deficiencies.  All forecasters have access to $(B_1,B_2)$ and
specify a Gumbel copula with Kendall's $\tau$ equal to $\hat{\tau}$
and normal marginals, where the first margin $F_1$ has mean
$\hat{\mu}_1$ and unit variance, and the second margin $F_2$ has mean
zero and variance $\hat{\sigma}_2^2$, with details provided in Table
\ref{tab:1}.  We name each forecaster with a sequence of three
letters, where T stands for true and F for false.  For example, the
forecaster TTF specifies the first and the second marginal
distributions correctly, but misspecifies the copula.  The forecaster
TTT is ideal with respect to the $\sigma$-algebra generated by
$(B_1,B_2)$ in the sense defined in Section \ref{sec:notions} and does
not show any forecast deficiencies.

\begin{figure}[p]
\centering
\includegraphics{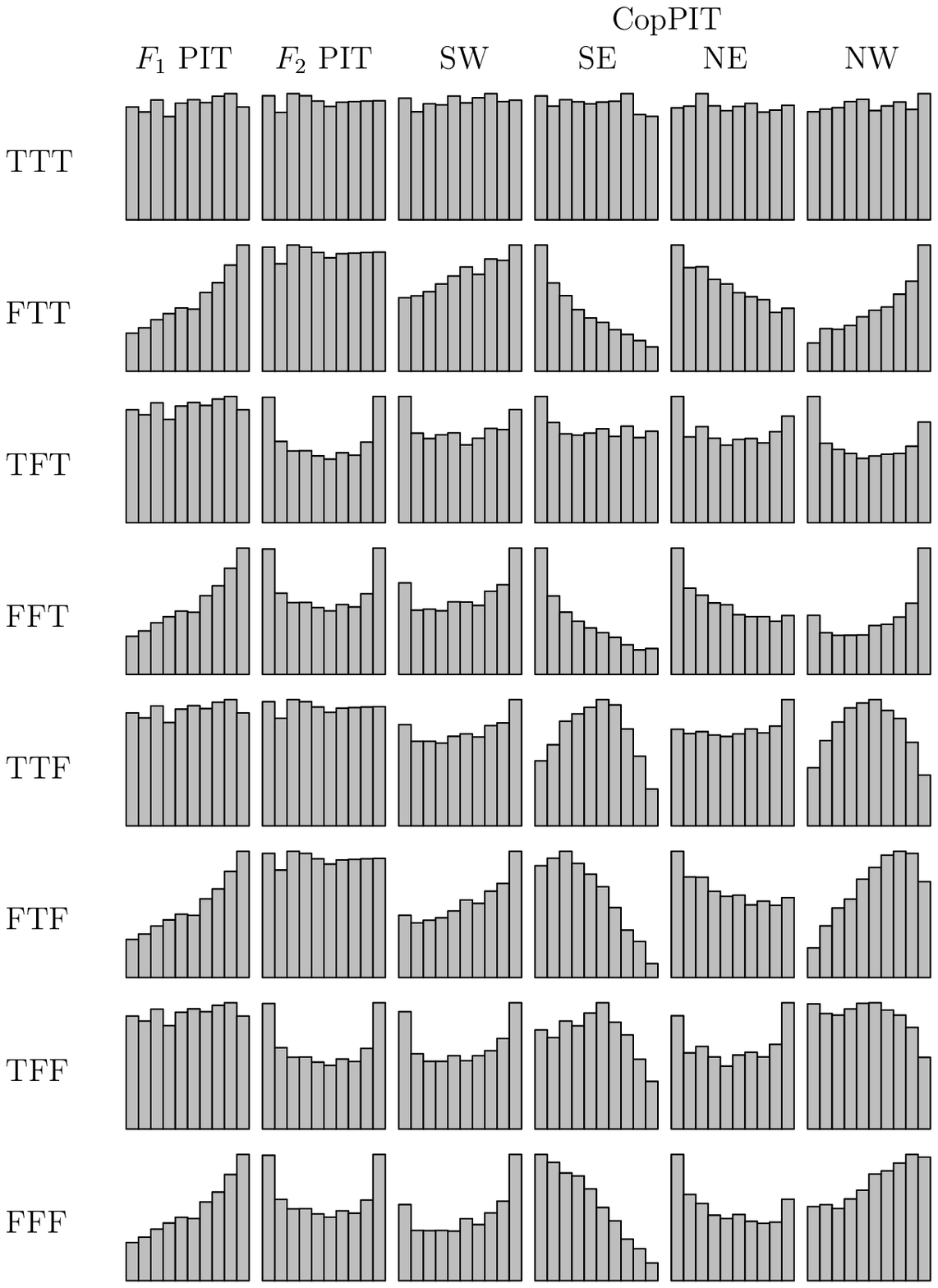}
\caption{Univariate PIT and directional CopPIT histograms for
  the forecasters in the simulation study \label{fig:3}}
\end{figure}

\begin{figure}[p]
\centering
\includegraphics[width=0.7\textwidth]{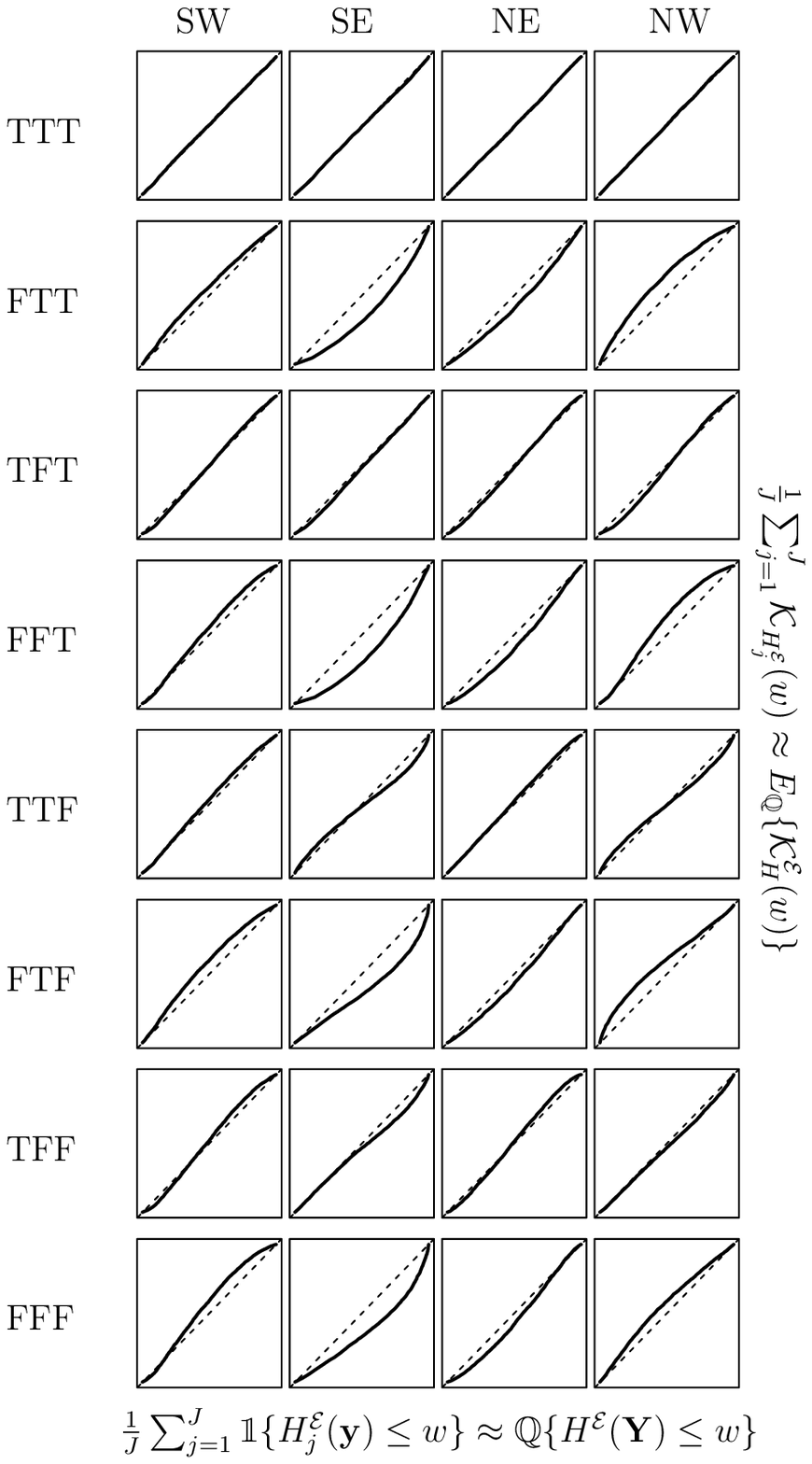}
\caption{Directional climatological calibration plots for the forecasters in
  the simulation study \label{fig:clical}}
\end{figure}

Figure \ref{fig:2} shows CopPIT histograms for the eight forecasters
based on a sample of $4,000$ forecast--observation pairs.  It is
interesting to observe that the standard CopPIT histogram detects
misspecified marginals as well as misspecified copulas.  Similar to
the interpretation of univariate PIT histograms
\citep{GneitingETAL2007}, biases yield skewed histograms,
underdispersed forecasts induce a U-shape, and overdispersed forecasts
an inverse U-shape.

Figure \ref{fig:3} shows univariate PIT histograms along with
directional CopPIT histograms based on another sample of $4,000$
forecast--observation pairs.  The joint consideration of the
histograms can diagnose specific forecast deficiencies.  As a rule of
thumb, the CopPIT histograms mimic features seen in the univariate PIT
histograms if the copula is well specified.  In contrast, if the
copula is ill specified, the CopPIT histograms show deviations from
uniformity in shapes that are not necessarily reflected by the PIT
histograms.  Finally, Figure \ref{fig:clical} shows directional
climatological copula calibration plots.  While misspecifications of
the probabilistic forecasts are readily discernible, the climatological
copula calibration plots appear to be more difficult to interpret
diagnostically than the CopPIT histograms.

\section{Case study: Probabilistic forecasts of wind vectors over the 
         Pacific Northwest}  \label{sec:EMOS}

In a recent change of paradigms, meteorologists have adopted
probabilistic weather forecasting in the form of ensemble forecasts.
An ensemble forecast is a collection of numerical weather prediction
(NWP) model runs that are based on distinct initial conditions and/or
model physics parameters \citep{GRaftery2005, Leutbecher2008}.
Despite their undisputed success, ensemble forecasts tend to be biased
and underdispersed, in the sense of the spread among the ensemble
members being too small to be realistic.  Therefore, methods for the
statistical postprocessing of ensemble forecasts have been developed,
such as the ensemble model output statistics (EMOS) approach of
\cite{GneitingETAL2005}, which generates Gaussian predictive
distributions for univariate variables.  In a more recent development,
\cite{SchuhenETAL2012} developed a bivariate EMOS method that
generates bivariate Gaussian predictive distributions for wind
vectors.

Here, we take up their work on probabilistic forecasts of surface wind
vectors over the North American Pacific Northwest based on the
University of Washington Mesoscale Ensemble \citep{Eckel2005}, which
has $m = 8$ members.  The test data comprise calendar year 2008 with a
total of 19,282 forecast--observations pairs at a prediction horizon
of 48 hours.  We assess and compare the raw ensemble forecast, the
statistically postprocessed regional bivariate EMOS forecast
developed by \cite{SchuhenETAL2012}, and an Independent EMOS forecast
with the same bivariate Gaussian predictive distribution, except that
the correlation coefficient is misspecified at zero.

\begin{figure}[t]
\centering
\includegraphics{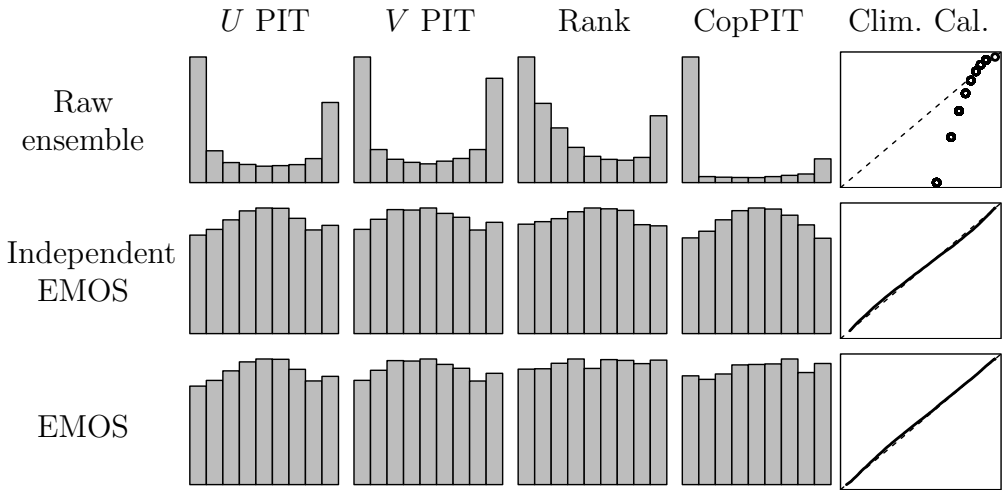}
\caption{Univariate PIT histograms, multivariate rank histogram,
  CopPIT histogram and climatological CopPIT calibration plot for
  the raw ensemble, Independent EMOS, and EMOS forecasts of wind
  vectors.  Following common practice, we label the wind vector
  components as $u$ and $v$. \label{fig:EMOS}}
\end{figure}

Figure \ref{fig:EMOS} shows univariate PIT histograms, the
multivariate rank histogram, the CopPIT histogram, and the
climatological copula calibration plot for the raw ensemble,
Independent EMOS, and EMOS forecasts.  The raw ensemble forecast shows
U-shaped PIT, multivariate rank and CopPIT histograms, which attest to
its underdispersion, and the climatological copula calibration plot
points at severe forecast deficiencies.  The univariate PIT histograms
for the Independent EMOS and EMOS forecasts are identical and diagnose
slight overdispersion.  However, the bivariate rank and CopPIT
histograms for the EMOS forecast are more uniform than for the
Independent EMOS forecast, as the Independent EMOS technique fails to
take dependencies between the wind vector components into account,
with the CopPIT histogram providing a much clearer diagnosis than the
multivariate rank histogram.

\section{Discussion}  \label{sec:discussion} 

In this paper, we introduced the copula probability integral transform
(CopPIT), and we proposed CopPIT histograms and climatological copula calibration
diagrams as diagnostic tools in the evaluation and comparison of
probabilistic forecasts of multivariate quantities.  These tools apply
to non-parametric, semi-parametric and parametric approaches and
thus can be employed to diagnose strengths and deficiencies of
multivariate stochastic models in nearly any setting, be it predictive
or not.

Extant methods for calibration checks for probabilistic forecasts of
multivariate quantities apply either to ensemble forecasts only, such
as the minimum spanning tree rank histogram and the multivariate rank
histogram \citep{SmithHansen2004, Wilks2004, GneitingETAL2008}, or
they apply to density forecasts only, such as the methods of
\cite{DieboldETAL1999}, \cite{Ishida2005}, and \cite{Gonzalez2012} that
rely on the univariate PIT and the Rosenblatt transform
\citep{Rosenblatt1952, Ruschendorf2009} in one way or another.  By way
of contrast, CopPIT histograms and climatological copula calibration diagrams
apply to all types of probabilistic forecasts, including both, ensemble
forecasts and density forecasts.

In our case study, we assessed probabilistic forecasts of raw ensemble
and statistically postprocessed density forecasts of bivariate wind
vectors.  However, our methods also apply in higher dimensions and
then it may be useful to plot CopPIT histograms and climatological
copula calibration diagrams for a range of subvectors of the outcome, too.

As noted, probabilistic forecasting strives to maximize the sharpness
of the predictive probability distributions subject to calibration
\citep{GneitingETAL2007}, and the methods proposed here serve to
evaluate calibration only.  If probabilistic forecasters are to be
ranked considering both calibration and sharpness, proper scoring
rules can be employed \citep{GRaftery2007, GneitingETAL2008}, with
recent theoretical advances having been made by \cite{ Ehm2012}.
\cite{DiksETAL2010} and \cite{RopnackETAL2013} advocate the use of the
logarithmic score to compare probabilistic forecasts of multivariate
quantities.  The event based approach of \citet{PinsonGirard2012}
reduces a high-dimensional quantity to a binary event --- essentially,
the ultimate dimension reduction --- and applies proper scoring rules
to assess the induced probability forecasts for dichotomous events.
While these techniques aim to rank probabilistic forecasters, CopPIT
histograms and climatological copula calibration diagrams are diagnostic tools
that strive to inform model development and spur model improvement.

\section*{Acknowledgements} 

The authors thank Nina Schuhen and Thordis Thorarinsdottir for
assistance with the data handling.  Tilmann Gneiting acknowledges
funding from the European Union Seventh Framework Programme under
grant agreement no.~290976.

\bibliographystyle{plainnat}
\bibliography{biblio}

\end{document}